\documentclass[letterpaper, 10 pt, conference]{ieeeconf} 
\IEEEoverridecommandlockouts
\usepackage{defs}

\usepackage{algorithm}
\usepackage{algpseudocodex}

\newcommand{\bbW}{\mathbb{W}}
\newcommand{\bbWh}{\widehat{\bbW}}
\newcommand{\capsize}{\footnotesize}

\newcommand{\hsom}{\hspace{1em}}

\newcommand{\smfm}[1]{\mbox{\footnotesize $#1$}}
\newcommand{\ssN}{{\scriptscriptstyle N}}
\newcommand{\ssM}{{\scriptscriptstyle M}}
\newcommand{\cM}{{\cal M}} 
\newcommand{\cN}{{\cal N}}
\newcommand{\tsty}{\textstyle} 
\newcommand{\sfT}{{\sf \scriptscriptstyle T}}
\newcommand{\sfS}{{\sf \scriptscriptstyle S}}

\newtheorem{problem}{Problem}

\title{\LARGE \bf {An Interpolation-based Scheme for \\ Rapid Frequency-Domain System Identification}}
\author{Jared Jonas and Bassam Bamieh$^1$%
\thanks{$^1$Jared Jonas and Bassam Bamieh are associated with the Department of Mechanical Engineering, University of California - Santa Barbara {\tt\small \{jjonas,bamieh\}@ucsb.edu}}}

\begin{document}
\maketitle
\thispagestyle{empty}
\pagestyle{empty}

\begin{abstract}
We present a frequency-domain system identification scheme based on barycentric interpolation and weight optimization. The scheme is related to the Adaptive Antoulas-Anderson (AAA) algorithm for model reduction, but uses an adaptive algorithm for selection of frequency points for interrogating the system response, as would be required in identification versus model reduction. The scheme is particularly suited for systems in which any one sinusoidal response run is long or expensive, and thus there is an incentive to reduce the total number of such runs. Two key features of our algorithm are the use of transient data in sinusoidal runs to both optimize the barycentric weights, and automated next-frequency selection on an adaptive grid. Both are done with error criteria that are proxies for a system's $\sf H^2$ and $\sf H^\infty$ norms respectively. Furthermore, the optimization problem we formulate is convex, and can optionally guarantee stability of the identified system.  Computational results on a high-order, lightly damped structural system highlights the efficacy of this scheme.  
\end{abstract}

\section{Introduction}

Typical frequency domain system identification algorithms proceed by injecting sinusoidal signals, recording the steady-state gain and phase response, and repeating such experiments over a large set of frequencies. One then obtains samples (in frequency) of the frequency response of the system. Several methods can then be used to ``fit'' a transfer function to this frequency-domain data, which include fitting methods such as  vector fitting~\cite{gustavsen2002rational}, least-squares-based fitting schemes such as~\cite{sanathanan2003transfer}, or the various subspace identification based methods~\cite{jamaludin2013n4sid}. There are also interpolation-type techniques such 
as Pade approximation, moment matching~\cite{astolfi2010model}, and the Loewner framework~\cite{ionita2014data,antoulas2017tutorial,karachalios2021loewner}, which are more closely related to the present paper. 

In this paper we present a new frequency-domain identification scheme that appears to be well-suited for high-order lightly-damped systems such as those that arise in acoustics, thermoacoustics, and light structures. Such systems have features that make frequency domain identification challenging in certain settings. First, the presence of lightly damped modes implies long transient times for each sinusoidal experiment. More importantly, for certain analysis and design problems, highly accurate gain and phase data is required around high-\(Q\) resonant modes~\cite{Epperlein15}, thus requiring a rather fine frequency grid around such resonances. 


The identification scheme presented here is an interpolation-based scheme inspired by the Adaptive Antoulas-Anderson (AAA) algorithm and its variants~\cite{Nakatsukasa_2018,Gosea24Stable,benner2021interpolation,Jonas2024}. Our scheme is motivated by requirements for frequency domain identification experiments where each sinusoidal run is either time consuming or difficult to do, and therefore there is an incentive to reduce their number. One example is in identification of thermoacoustic dynamics such as those in the Rijke tube~\cite{Epperlein15}. The scheme contains the following features. (i) It is based on barycentric interpolation, which gives exact interpolation at given frequency points, yet allows for weight selection to optimize the response at all other frequencies. (ii) We use the transient data from each sinusoidal run to formulate an optimization problem for the weights, exploiting the natural idea that those transients contain information about the overall frequency response. (iii) We interpolate at the frequency from an adaptive grid which has the highest approximation error. 


This paper is organized as follows. Some background on state-space based barycentric interpolation and the overall problem formulation is given in the next section. Section~\ref{main.sec} presents the algorithmic details of weight optimization, and the main stability result. Numerical results from a high-order lightly damped structural system are presented in Section~\ref{computation.sec}, and we end with some concluding remarks. 

\section{Background and Problem Formulation}

In this paper we consider Single-Input-Single-Output system, although the ideas are readily generalizable to the Multi-Input-Multi-Output case. Consider the following rational function
\begin{equation}
\begin{aligned} 
	R(s) &= M\inv(s)  \, N(s), \\
    M(s) &= I + \sum_{k=0}^\ell \tfrac{w_k}{s-j\omega_k},  ~
    N(s) = D + \sum_{k=0}^\ell \tfrac{w_k \, G(j\omega_k)}{s-j\omega_k}, 			
\end{aligned} \label{eq:M_N_def}
\end{equation}
where for each $k$,  \(\omega_k\) represents a distinct frequency, and  \(w_k\) is a non-zero scalar weight term.  The rational function \(R\) is known as a \textit{barycentric interpolant}.  It has the following interpolation properties~\cite{Jonas2025TAC}:
\begin{enumerate}
    \item For \(w_k \neq 0\), \(R(j\omega_k) = G(j\omega_k)\).
    \item \(\lim_{\omega\to\infty} R(j\omega) = D\).
\end{enumerate}

Barycentric interpolants like~\eqref{eq:M_N_def} are useful in problems where exact interpolation at certain frequencies is required, but the additional freedom in selecting the weights $\{ w_k \}$ allows for optimizing other error criteria. This framework was used for example in~\cite{Jonas2025TAC} for model reduction where  weights are selected to minimize a proxy measure for an $\sf H^2$ system error norm.  

Special realizations of $M,N$ above which combine complex-conjugate poles and interpolation points are used to ensure the resulting system has real coefficients.  Let \(D\), \(G(0)\), and $\left\{\big( \omega_k , G(j\omega_k) \big) \right\}_{k=1}^\ell$ be given interpolation data, and let \(M\) and \(N\) be the following systems
\begin{equation} 
        \begin{aligned}
            M(s) &=\tsty I + w_0\,  \mathcal{M}_0(s) + \sum_{k=1}^\ell \, w_k \,  \mathcal{M}_k(s), 
            																		\\
            N(s) &= \tsty D + w_0\,  \mathcal{N}_0(s) + \sum_{k=1}^\ell \, w_k\,  \mathcal{N}_k(s), 
        \end{aligned}								
  \label{eq:M_N_def2} 
\end{equation}
where \(w_0\), \(\cM_0\), and \(\cN_0\) represent the interpolation point at $\omega=0$.  Realizations with real parameters for each \(\mathcal{M}_k\) and \(\mathcal{N}_k\) are 
\[
	\mathcal{M}_k = \smfm{ \brac{\begin{array}{c|c}\mathcal{A} & \mathcal{B}_{\ssM, k} \\ \hline I & 0\end{array}}}, 
		\hsom  \smfm{ \mathcal{N}_k = \brac{\begin{array}{c|c}\mathcal{A} & \mathcal{B}_{\ssN, k} \\ \hline I & 0\end{array}}},
\]
with the matrices
\begin{align*}
    \mathcal{A}_0 &= 0, \;   & \mathcal{B}_{\ssM, 0} &= 1,  \;  &  \mathcal{B}_{\ssN, 0} &= G(0),				\\
    	\mathcal{A}_k &=  \smfm{\begin{bmatrix}0 & \omega_k \\ -\omega_k & 0\end{bmatrix}}, 
  	 \; & \mathcal{B}_{\ssM, k} &=  \smfm{ \begin{bmatrix}1 \\ 0\end{bmatrix}}, 
  	  \; & \mathcal{B}_{\ssN, k} &=  \smfm{\begin{bmatrix}\Re(G(j\omega_k)) \\ -\Im(G(j\omega_k))\end{bmatrix}}.
\end{align*}

The weights in~\eqref{eq:M_N_def2} can be factored out, and $M$ and $N$ can be represented in terms of weight-independent systems \(\cN\), \(\cM\), and a matrix \(\bbW\) made up of the weights.  Define these two SIMO systems \(\mathcal{M}\) and \(\mathcal{N}\)
\begin{align*}
    \mathcal{M} &= \begin{bsmallmatrix}I \\ \mathcal{M}_0 \\ : \\ \mathcal{M}_\ell\end{bsmallmatrix} 
    	= \smfm{ \brac{\begin{array}{c|c}\mathcal{A} & \mathcal{B}_\ssM \\ \hline 0 & 1 \\ I & 0\end{array}} },
    \; \mathcal{N} = \begin{bsmallmatrix}D \\ \mathcal{N}_0 \\ : \\ \mathcal{N}_\ell\end{bsmallmatrix} 
    	= \smfm{\brac{\begin{array}{c|c}\mathcal{A} & \mathcal{B}_\ssN \\ \hline 0 & D \\ I & 0\end{array}}}, 
\end{align*}
where 
\begin{align*}
    \mathcal{A} = \mathsf{blkdiag} \p{{\mathcal{A}_0, \ldots, \mathcal{A}_\ell}}, ~
    \mathcal{B}_\ssM = \begin{bsmallmatrix}\mathcal{B}_{M,0} \\ : \\ \mathcal{B}_{\ssM,\ell}\end{bsmallmatrix}, ~
    	\mathcal{B}_\ssN = \begin{bsmallmatrix}\mathcal{B}_{N,0} \\ : \\ \mathcal{B}_{N,\ell}\end{bsmallmatrix}.
\end{align*}
Then we see $R$ defined in~\eqref{eq:M_N_def} can be rewritten in terms of these systems \(\cM\), \(\cN\), and a matrix \(\bbW\) containing the weights as 
\[
	R = (\bbW \mathcal{M})\inv (\bbW \mathcal{N}), \hsom 
		\bbW := \left[ \arraycolsep=2pt \begin{array}{c:ccc} 1 & w_0 & \cdots & w_\ell \end{array} \right] 
		=: \big[ \arraycolsep=2pt \begin{array}{c:c} 1 & \bbWh \end{array} \big] .
\]
A state space representation for $R$ is given by 
\begin{equation}
    R  = \smfm{
    	\brac{\begin{array}{c|c}\mathcal{A} - \mathcal{B}_\ssM \bbWh & \mathcal{B}_M D - \mathcal{B}_N \\ \hline -\bbWh & D \end{array}}}, \label{eq:R}
\end{equation} 
Details of these derivation can be found in~\cite{Jonas2025TAC}.  

In~\cite{Jonas2025TAC} we used the above setting for interpolation-based model reduction of a high-order system $G$. In this paper $G$ is unknown, and the interpolation data $\left\{\big( \omega_k , G(j\omega_k) \big) \right\}_{k=0}^\ell$ will come from frequency-domain system identification experiments. 

\subsection{Problem formulation}

Given an unknown Linear Time Invariant (LTI) system and a frequency point $\omega_k$, a typical frequency-domain identification experiment uses a pure sinusoid\footnote{There are of course other choices such as chirp or white noise inputs. However, we are interested in setting where very accurate measurements of the frequency response at special frequencies is required. This normally requires a pure sinusoid input.} of frequency $\omega_k$ as input to this system.  The response $y^{(k)}$ is recorded. The response can be divided into two segments $(y^{(k)}_\sfT, y^{(k)}_\sfS)$, where $y_\sfT$ and $y_\sfS$ denote the transient and the steady state portion of the response. Detection of the point of time that separates the two portions of the response can be done in many ways, one of which is described in Appendix~\ref{transient_steady.app}.  An outline for our algorithm is:
\begin{itemize} 
	\item After collecting the data from run \(k\), the steady state part $y^{(k)}_\sfS$ is used to determine the complex value $G(j\omega_k)$. This is then used as one interpolation point in~\eqref{eq:M_N_def}. 
	\item There is information in the transient portion $y^{(k)}_\sfT$ of the response. This can be used in an optimization problem to select the free weights $\{ w_k \}$ in~\eqref{eq:M_N_def}. The optimization problems we formulate is stated in Theorems~\ref{thm:sol} and~\ref{thm:stable}. The objective is a proxy for the error between the system model acting on all the previous input data, compared to the actual measured output data. The convex optimization problem in theorem~\ref{thm:stable} is formulated to guarantee stability of the identified model. 
	\item We form an adaptive grid of interpolation points by running experiments at frequencies between the ones we are already interpolating, then interpolating at the one with highest approximation error.
\end{itemize} 

The next section describes the algorithm and the various choices involved in more detail. 

\section{Algorithms and Main Results} \label{main.sec}

\subsection{System ID experiments} \label{sec:sys_id}

Consider the unknown system \(G\).  In order to construct the interpolant~\eqref{eq:M_N_def}, we need estimations for the feedthrough term \(D\), the DC gain \(G(0)\), and the frequency response for each frequency we're interpolating.  Consider the \(k\)th system identification experiment and the desired frequency \(\omega_k\).  We follow the procedure outlined in the appendix section~\ref{transient_steady.app} to collect the transient data \((u^{(k)}, \, y^{(k)})\) and estimated frequency response \(G(j\omega_k)\).  

\subsection{Weight optimization}

From the system ID experiments, we have a collection of transient data we utilize for setting the weighting free parameters of the system.  We want to choose the weights to replicate the outputs of the system \(R\) subject to the same \(m\) sinusoidal inputs.  In other words, we want to minimize the sum
\[\sum_{k=1}^m \frac{1}{n_k} \norm{\hat{y}^{(k)} - y^{(k)}}_2^2,\]
where \(\hat{y}^{(k)}\in\real^{n_k}\) is the output of \(R\) subject to the input \(u_k\), and \(y^{(k)}\in\real^{n_k}\) is the recorded output of the unknown system \(G\).  Essentially we are summing the mean squared error of each of the \(m\) experiments.  However, a minimization of this sum over the set of weights is non-convex in \(\bbW\), and typically this is handled by instead minimizing a related sum which has been weighted by \(M\).  Thus we consider the minimization of the quantity
\[\sum_{k=1}^m \frac{1}{n_k} \norm{M \p{\hat{y}^{(k)} - y^{(k)}}}^2,\]
which we rewrite as
\[\sum_{k=1}^m \frac{1}{n_k} \norm{Nu^{(k)} - My^{(k)}}^2\]
owing to the fact that \(M\hat{y}^{(k)} = MRu^{(k)} = Nu^{(k)}\).  Factoring out the weights yields the optimization problem shown in problem~\ref{prob:opt}, which has an explicit solution detailed in theorem~\ref{thm:sol}.

\begin{problem} \label{prob:opt}
\[\min_{\bbWh} \sum_{k=1}^m \frac{1}{n_k} \norm{\begin{bsmallmatrix}1 & \widehat{\bbW}\end{bsmallmatrix}\p{\mathcal{N}u^{(k)} - \mathcal{M}y^{(k)}}}^2,\]
with \(\mathcal{M}\) and \(\mathcal{N}\) defined in the background section.
\end{problem}

\begin{theorem} \label{thm:sol}
Consider the set of \(m\) input-output pairs \(\curly{(u^{(k)}, \, y^{(k)})}_{k=1}^m\) where \(u^{(k)}, \, y^{(k)}\in\real^{n_k}\) and form the empirical covariance matrix
\[X\in\real^{(2\ell+2)\times(2\ell+2)} := \sum_{k=1}^m X_k,\]
where \(X_k\) is the covariance matrix of the signal \(x^{(k)} := \mathcal{N} u^{(k)} - \mathcal{M}y^{(k)}\), i.e. \(X_k := \mathrm{cov}\p{x^{(k)}, \, x^{(k)}}\) where \(\mathrm{cov}\p{x, \, z} := \frac{1}{n} \sum_{i=1}^n x_i z_i^*\) and \(x\), \(z\) have \(n\) samples.  Then partition \(X\) into
\[X =: \begin{bsmallmatrix}\hat{X}_1 & \hat{X}_0 \\ \hat{X}_0^* & \hat{X}_2\end{bsmallmatrix}, \; \hat{X}_1\in \real.\]
If \(\hat{X}_2\) is positive definite, then the objective function in the optimization problem~\ref{prob:opt} has the minimizer
\[\bbWh = -\hat{X}_0\hat{X}_2\inv.\]
\end{theorem}

\begin{proof}
First, rewrite the squared norm in problem~\ref{prob:opt} as
\[\begin{bsmallmatrix}1 & \widehat{\bbW}\end{bsmallmatrix} X_k \begin{bsmallmatrix}1 \\ \widehat{\bbW}^*\end{bsmallmatrix}.\]
Factor out the weight vectors from the sum to get
\[\min_\bbW \begin{bsmallmatrix}1 & \widehat{\bbW}\end{bsmallmatrix} X \begin{bsmallmatrix}1 \\ \widehat{\bbW}^*\end{bsmallmatrix}.\]
Now, partition \(X\) such that it conformably multiplies the weight vector, yielding the minimization
\[\min_{\widehat{\bbW}} \; \begin{bsmallmatrix}1 & \widehat{\bbW}\end{bsmallmatrix} \begin{bsmallmatrix}\hat{X}_1 & \hat{X}_0 \\ \hat{X}_0^* & \hat{X}_2\end{bsmallmatrix} \begin{bsmallmatrix}1 \\ \widehat{\bbW}^*\end{bsmallmatrix}.\]
From our previous work~\cite{Jonas2025TAC}, we know if \(\hat{X}_2\) is positive definite and if \(\hat{X}_0\) is in the column space of \(\hat{X}_2\), then the minimization has the solution
\[\widehat{\bbW} = - X_0 X_2\inv.\]
If \(\hat{X}_2\) is positive definite, then the column vector \(\hat{X}_0\) is always in the column space of \(\hat{X}_2\), thus the positive definiteness of \(\hat{X}_2\) is the only requirement.  
\end{proof}

\subsection{Enforcing stability}
Unfortunately, the weights generated from the previous optimization problem often don't yield a stable system \(R\).  In many cases we want \(R\) to be stable, thus we will develop a modified optimization problem which enforces the stability of the resulting system.  In other words, we want to solve 
\begin{align}
    &\min_{\bbWh} \sum_{k=1}^m \frac{1}{n_k} \norm{\begin{bsmallmatrix}1 & \widehat{\bbW}\end{bsmallmatrix}\p{\mathcal{N}u^{(k)} - \mathcal{M}y^{(k)}}}^2, \label{eq:min_stable} \\ 
    &R \; \mathrm{stable}, \nonumber
\end{align}
ideally with a convex solver.  In theorem~\ref{thm:stable}, we derive a convex relaxation problem~\ref{prob:lmi} from problem~\ref{prob:opt} which guarantees stability of the resulting system.

\begin{problem} \label{prob:lmi}
\begin{align*}
        &\min_{P, \, Q, \, \gamma} \gamma,  \\ 
        & P > 0, \; \begin{bmatrix}\gamma & Q \\ Q^* & 2P - \hat{X}_2\end{bmatrix} > 0, \\
        & YP- \mathcal{B}_MQ + PY^* - Q^* \mathcal{B}_M^* < -2\alpha P, 
\end{align*}
where
\[Y = \mathcal{A} - \mathcal{B}_M \hat{X}_0 \hat{X}_2\inv, \; \alpha > 0,\]
and with \(\hat{X}_0\), \(\hat{X}_2\) as defined in theorem~\ref{thm:sol} and \(\mathcal{A}\), \(\mathcal{B}_M\) as defined in the background section.
\end{problem}

\begin{theorem} \label{thm:stable}
    Let \(Q\), \(P\), and \(\gamma\) be minimizers of the LMI problem~\ref{prob:lmi}.  Then, the system \(R\) is stable, where
    \[R = \smfm{\brac{\begin{array}{c|c}\mathcal{A} - \mathcal{B}_\ssM \bbWh & \mathcal{B}_\ssM D - \mathcal{B}_\ssN \\ \hline -\bbWh & D \end{array}}}\]
    and \(\bbWh = QP\inv - \hat{X}_0 \hat{X}_2\inv\).  Additionally, the cost function in the original problem~\ref{prob:opt} with the minimizer \(\bbWh\) is bounded above by \(\gamma - \hat{X}_0 \hat{X}_2\inv \hat{X}_0^* + X_1\).  
\end{theorem}

\begin{proof}
Consider the optimization problem~\eqref{eq:min_stable}. We can optimize this objective function such that the resulting system's \(A\) matrix satisfies an inequality which enforces the real part of the poles to be less than some \(-\alpha\), i.e.
\begin{align*}
    &\min_{\widehat{\bbW},\,P} \bbW X \bbW^*, \\ 
    &P > 0, \; (\mathcal{A}-\mathcal{B}_M \bbW)P+P(\mathcal{A}-\mathcal{B}_M \bbW)^* < -2\alpha P.
\end{align*}

\bigskip
First, we will rewrite the problem have a strictly quadratic objective in terms of our optimization variables.  We will substitute in our expression for \(\bbW\) and partition \(X\), giving the objective
\begin{align*}
    \bbW X \bbW^* &=  \begin{bsmallmatrix}1 & \widehat{\bbW}\end{bsmallmatrix} \begin{bsmallmatrix}\hat{X}_1 & \hat{X}_0 \\ \hat{X}_0^* & \hat{X}_2\end{bsmallmatrix} \begin{bsmallmatrix}1 \\ \widehat{\bbW}^*\end{bsmallmatrix} \\
    &= \hat{X}_1 + \hat{X}_0 \bbWh^* + \hat{X}_0^* \bbWh + \bbWh \hat{X}_2 \bbWh^*.
\end{align*}
Completing the square gives
\[(\bbWh+\hat{X}_0 \hat{X}_2\inv) \hat{X}_2 (\bbWh+\hat{X}_0 \hat{X}_2\inv)^* + \hat{X}_1 - \hat{X}_0 \hat{X}_2\inv \hat{X}_0,\]
Which then becomes \(Z \hat{X}_2 Z^*\) after letting \(Z = \bbWh + \hat{X}_0 \hat{X}_2\inv\) and dropping the constant terms.  Making this substitution into the inequality yields
\[(Y - \mathcal{B} _M Z) P + P(Y - \mathcal{B}_M Z)^* < -2\alpha P,\]
where \(Y = \mathcal{A} - \mathcal{B}_M \hat{X}_0 \hat{X}_2\inv\), thus our minimization becomes
\begin{align*}
    &\min_{Z,\, P} Z\hat{X}_2 Z^*, \\
    &P > 0, \; (Y - \mathcal{B}_M Z) P + P(Y - \mathcal{B}_M Z)^* < -2\alpha P.
\end{align*}

\bigskip
The inequality is bilinear in \(Z\) and \(P\), so we make the substitution \(Q = ZP\) and optimize over \(P\) and \(Q\) instead, giving 
\begin{align*}
    &\min_{Q,\, P} QP\inv \hat{X}_2 P\inv Q^*, \\
    &P > 0, \; YP- \mathcal{B}_MQ + PY^* - Q^* \mathcal{B}_M^* < -2\alpha P.
\end{align*}

\bigskip
Now we will focus on the objective function \(QP\inv \hat{X}_2 P\inv Q^*\), which can't be incorporated into an LMI directly. Thus, we must bound it above by a conservative quantity that can.  Inspired by~\cite{Gosea24Stable}, we use a special case of Young's relation.  This relation states for positive definite \(S\) and \(P\),~\cite{Caverly24}
\[2P \leq PS\inv P + S.\]
In our case, we let \(S = \hat{X}_2\), giving \(P \hat{X}_2\inv P \geq 2P-\hat{X}_2\).  Taking the inverse of both sides and multiplying from the left and right by \(Q\) and \(Q^*\) yields
\begin{equation}
    Q(2P-\hat{X}_2)\inv Q^* \geq QP\inv \hat{X}_2 P\inv Q^*. \label{eq:opt_comp}
\end{equation}
This shows that our objective function is bounded above by this new quantity.  If we now introduce a scalar slack variable \(\gamma\) that upper bounds this quantity \(\gamma > Q (2P-\hat{X}_2)\inv Q^*\), then perform a Schur complement, we get the LMI
\[\begin{bmatrix}\gamma & Q \\ Q^* & 2P - \hat{X}_2\end{bmatrix} > 0.\]
Incorporating everything together gives us our resulting set of LMIs.  From equation~\ref{eq:opt_comp}, we substitute and expand to get the inequality
\[\gamma > \bbWh\hat{X}_2 \bbWh^* + \hat{X}_0\bbWh^* + \bbWh \hat{X}_0^*+\hat{X}_0 \hat{X}_2\inv \hat{X}_0^*.\]
Thus,
\[\gamma - X_0X_2\inv X_0^* + X_1 > \bbWh \hat{X}_2 \bbWh^* + \hat{X}_0\bbWh^* + \bbWh \hat{X}_0^* + \hat{X}_1.\]
The right-hand side of the inequality exactly equals the expanded cost function of problem~\ref{prob:opt}, thus the left-hand side serves as an upper bound for the original optimization problem.   
\end{proof}

\subsection{Frequency selection}
The last important step to consider is the frequency selection strategy.  In other AAA-like algorithms, an interpolation point is added at the frequency where the error is highest.  However, we don't have full access to the frequency response of the target system, only to the frequency response data we obtain from experiments.  Thus, the simplest option is to consider a frequency range of interest for the system \([\omega_{\min}, \, \omega_{\max}]\) and run a number of system identification experiments along a grid of equally-spaced frequencies. Then, interpolate the recorded frequency response and optimize the weights over all of the input-output data.  The problem with this however, is that not all interpolation points will appreciably improve the approximation, forcing the grid size to be very fine.  

We will instead consider an adaptive gridding strategy.  The main idea behind the adaptive strategy is maintaing a list of ``test frequencies'' and interpolating at the test frequency with the largest error at each step.  We start by interpolating both at the lowest and highest frequencies \(\omega_{\min}\) and \(\omega_{\max}\).  We maintain a store of data from system ID experiments performed at the logarithmic midpoints (i.e. the geometric mean) between each of the interpolated frequencies, and add the one with highest approximation error at each step.  We then perform two new system ID experiments at the geometric mean between the chosen frequency and its neighboring frequencies.  Pseudocode describing more details of this adaptive strategy is listed in algorithm~\ref{alg:adaptive}. 

\begin{algorithm}
    \caption{Adaptive frequency selection}
    \label{alg:adaptive}
    \begin{algorithmic}[1]
        \small
        \Require{Unknown system \(G\), feedthrough \(D\), DC gain \(K\), \(\omega_{\min}\), \(\omega_{\max}\)}
        \State The set of interpolation frequencies \(\curly{\omega_i} \gets \curly{\omega_{\min}, \, \omega_{\max}}\).
        \State The test frequency list \(\curly{\hat{\omega}_i}\gets\curly{\sqrt{\omega_{\min}\omega_{\max}}}\).
        \State Run system ID experiments at \(\omega_{\min}\), \(\omega_{\max}\), and \(\hat{\omega}_1\) to initialize input-output data store \(\curly{(u_i, y_i)}\) and frequency response data \(\curly{G(j\omega_i)}\) as per section~\ref{sec:sys_id}.
        \State Find \(\bbW\) using theorem~\ref{thm:sol} or theorem~\ref{thm:stable} if the stability of \(R\) is important.
        \State Use \(\bbW\), frequency response data, \(D\), and \(K\) to construct \(R\) from equation~\ref{eq:R}.
        \Repeat
            \State Choose the test frequency \(\hat{\omega}_k\) that has highest approximation error \(|R(j\hat{\omega}_k) - G(j\hat{\omega}_k)|\).
            \State Let \(\omega_l\) and \(\omega_h\) be the next lowest and highest frequency in \(\curly{\omega_i}\).
            \State Remove \(\hat{\omega}_k\) from test frequency list, add to interpolation frequency list.
            \State Add \(\curly{\sqrt{\hat{\omega}_k\omega_l}, \, \sqrt{\hat{\omega}_k\omega_h}}\) to the test frequency list.
            \State Perform two system ID experiments at these frequencies and store input-output an frequency response data.
            \State Find \(\bbW\) using the input-output data from \textit{all} experiments.
            \State Construct \(R\).
        \Until{the model \(R\) is satisfactory}
    \end{algorithmic}
\end{algorithm}

Overall, we have described two frequency selection strategies and introduced two weight optimization problems, which we will evaluate in the following computational results section.  

\section{Computational results}	\label{computation.sec}

In this section we will demonstrate the performance of the our algorithm with numerical examples.  First, we will compare the  performance of both frequency selection approaches by comparing Bode plots and their respective \(\sf H^2\) error norm for various system sizes.  Then, we discuss the effect of adding the stability constraint to the optimization problem.  In the following numerical examples, we use the \((1,\,1)\) channel of a 270-state ``ISS'' model, which describes the flexural dynamics of one of the modules of the International Space Station~\cite{iss_model}.  This model is treated as an unknown system which we can interrogate with a sinusoidal input in the frequency range \([0.5,\,90]\) Hz.  We then generate the resulting estimated model of the system using the approaches in the previous section, opting for the stabilizing optimization problem unless otherwise specified.  

We first investigate the qualitative differences between the gridded and adaptive frequency selection approaches.  Figures~\ref{fig:qual_simple} and~\ref{fig:qual_adaptive} show Bode plots of the ISS model and a 43-state system (21 interpolation points) generated using each approach respectively.  In figure~\ref{fig:qual_simple}, we see the dynamics at low frequencies are captured well, but the dynamics at higher frequencies are not modeled particularly well because the scale of the dynamics is much finer than the interpolation frequency grid.  In comparison, figure~\ref{fig:qual_adaptive} captures the low frequency dynamics just as well with fewer points, and resolves the middle and high frequencies much better; the finer grid helps resolve the details that are missed in the gridded approach, which is especially evident when viewing the phase plot.  Interestingly, the dynamics near the peaks are generally resolved well even though the nearest interpolated frequency is not particularly close.  This is due to the large effect the peaks have on the transient responses in the input-output data and highlights one of the benefits of this method.  

\begin{figure}[ht]
    \includegraphics[width=\linewidth]{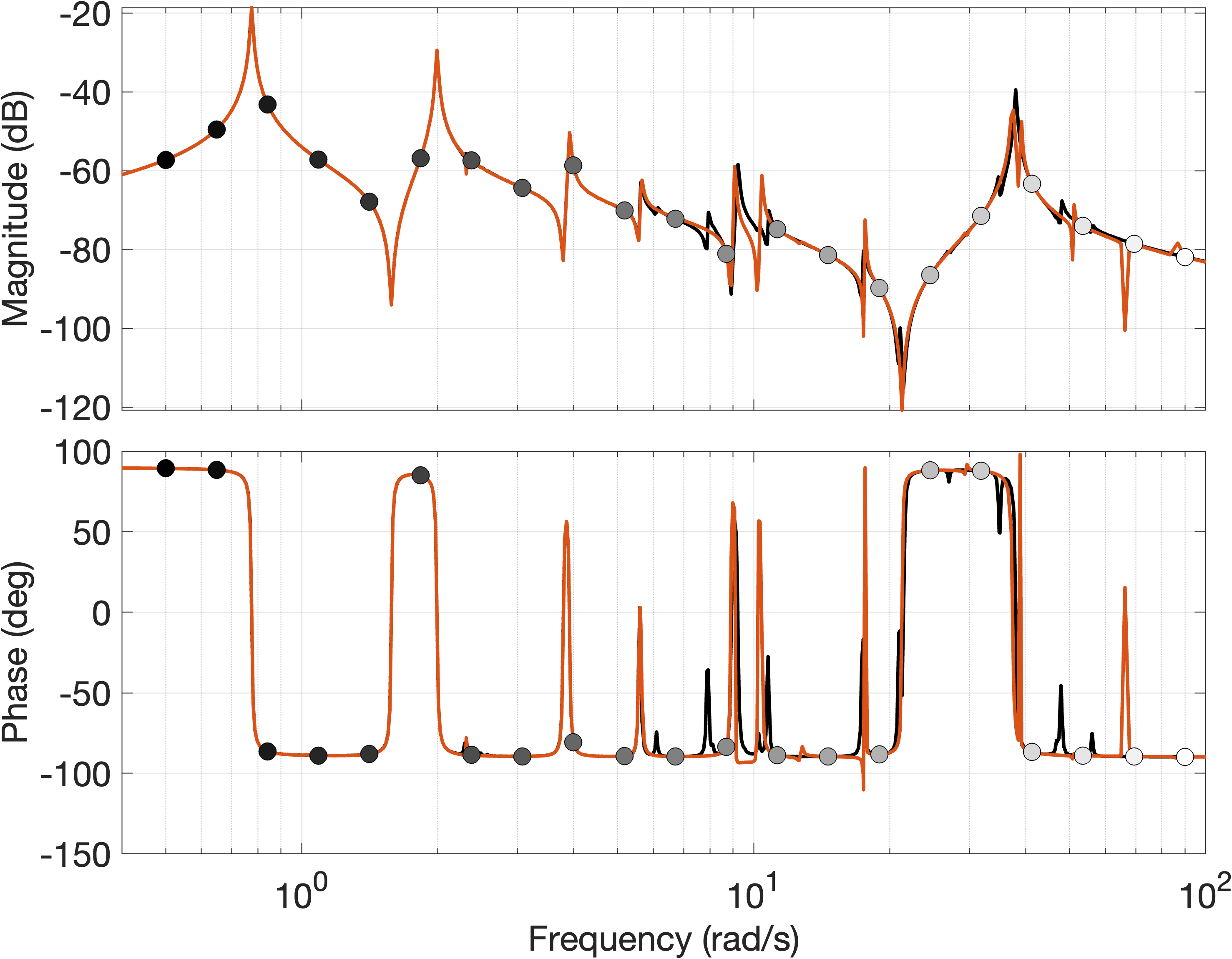}
    \caption{\capsize A Bode plot of the input ISS system (in black) and resulting 43 order system constructed using the gridded frequency selection approach (in red).  The circles indicate the frequency at which a system ID experiment was ran as well as the measured response data.}
    \label{fig:qual_simple}
\end{figure}

\begin{figure}[ht]
    \includegraphics[width=\linewidth]{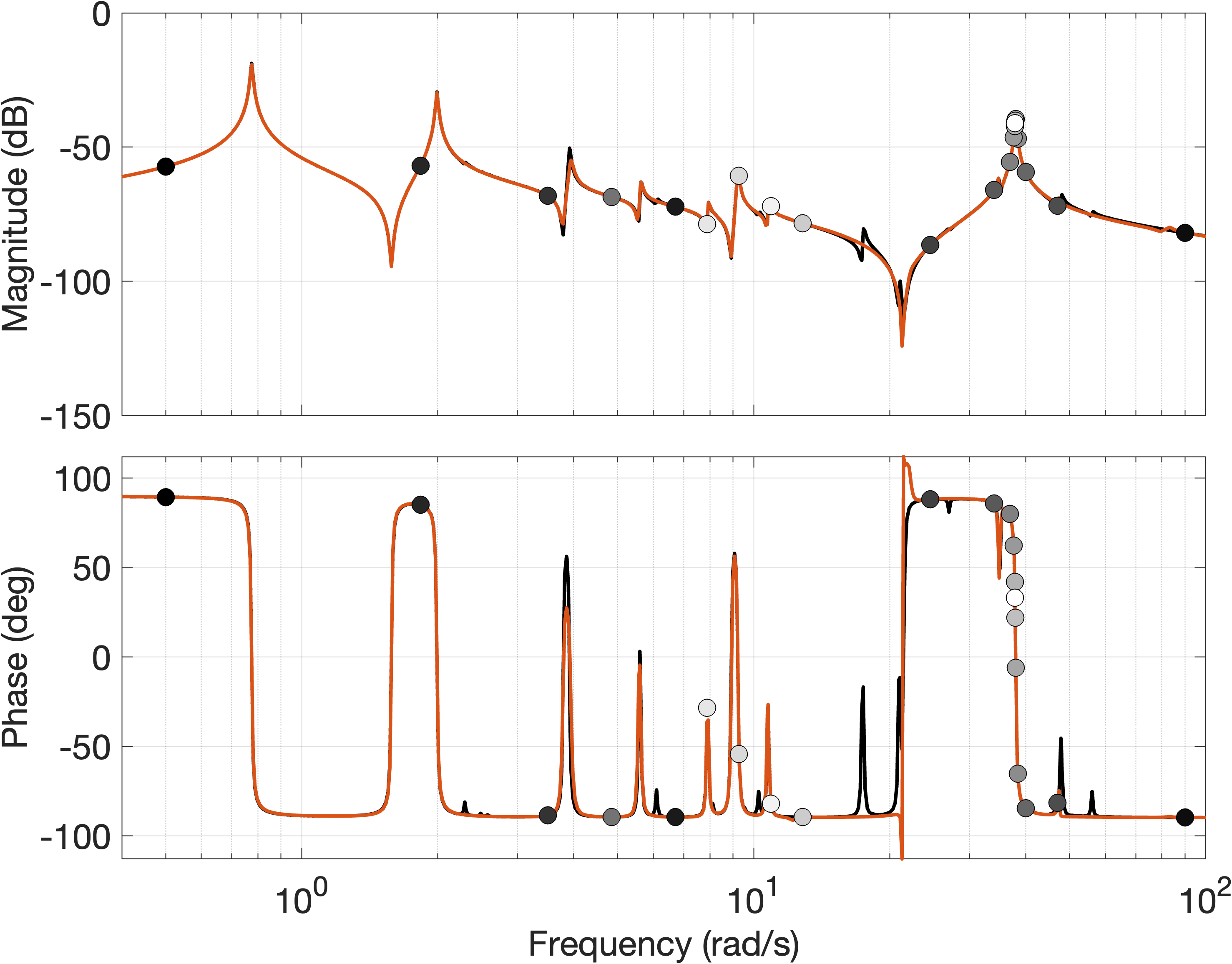}
    \caption{\capsize A Bode plot of the input ISS system (in black) and resulting 43 order system constructed using the adaptive frequency selection approach (in red).  The circles indicate the frequency at which a system ID experiment was ran as well as the measured response data.  The circle color indicates the order the experiments were ran, with black being oldest and white being newest.}
    \label{fig:qual_adaptive}
\end{figure}

Figure~\ref{fig:norm_methods} highlights the \(\sf H^2\) error norm of both approaches for a varying number of poles.  The performance of the gridded approach doesn't show a clear pattern as more poles are added; it would be hard to know a priori how many interpolation points would be required to achieve a satisfactory system.  In contrast, the adaptive approach demonstrates improvement as the system size grows without much variation.  On the whole, its performance is superior to that of the gridded approach as it achieves a lower error norm.

\begin{figure}[ht]
    \includegraphics[width=\linewidth]{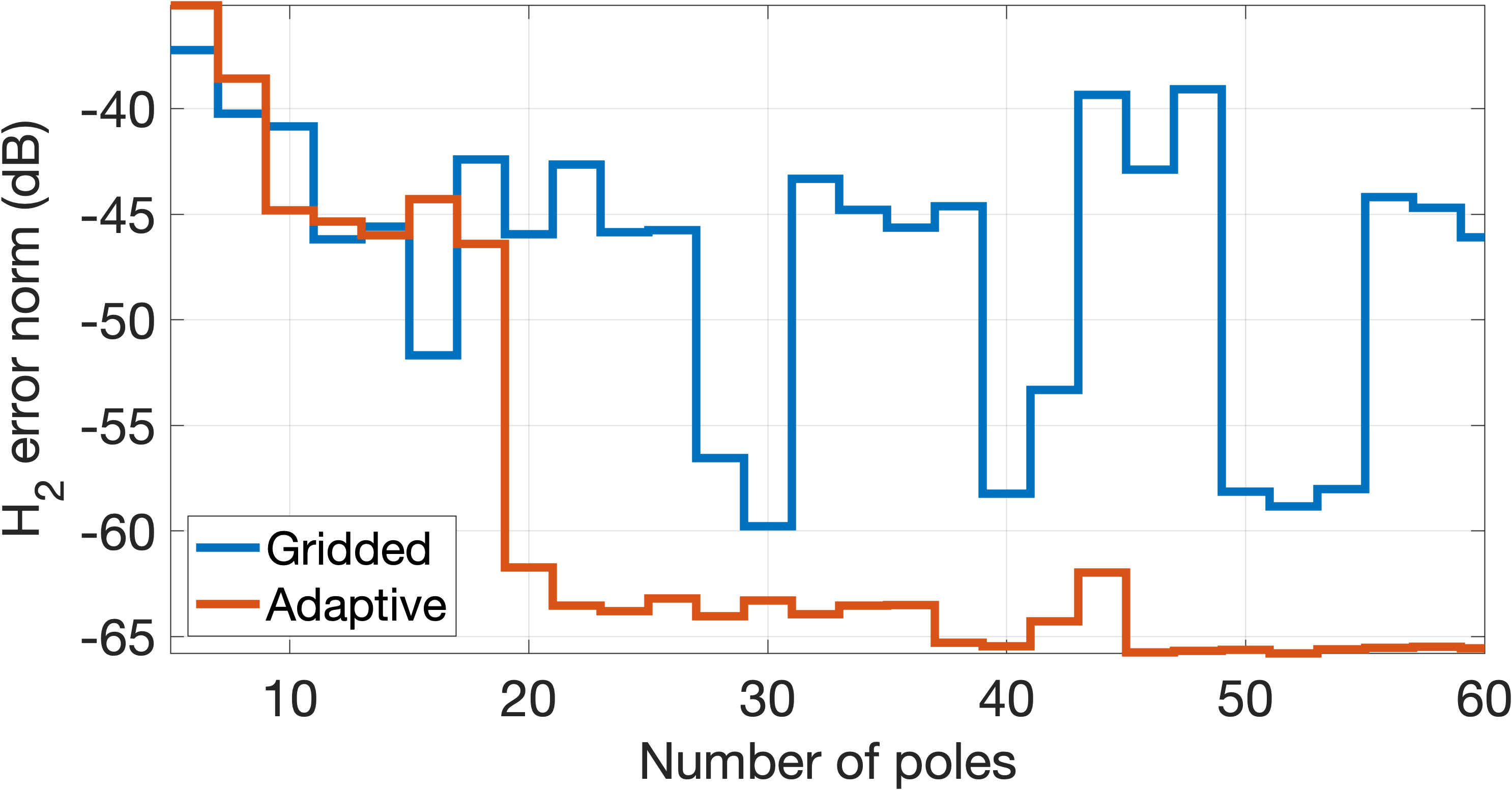}
    \caption{\capsize A plot showing the \(\sf H^2\) norm of the error system, i.e. \(R-G\), against the number of poles of each system \(R\) generated from the two frequency selection strategies with the input ISS system \(G\).  The blue line and red line indicate the error norm of systems generated with the gridded frequency selection strategy and the adaptive frequency selection strategy respectively.}
    \label{fig:norm_methods}
\end{figure}

Finally, we show the effect of adding the stability constraint in figures~\ref{fig:norm_opt} and~\ref{fig:bode_stability}.  We see with the exception of the 30-pole system, the \(L_\infty\) norm of the error system is affected negligibly by the addition of the stability constraint and resulting relaxed optimization problem.  This exception is caused by the appearance of a lightly-damped pole that disappears in the next iteration and has no discernable effect on the \(\sf H^2\) error norm.  Looking at the Bode plot when the system size is 11 in figure~\ref{fig:bode_stability}, we see that the magnitude response indeed matches well, and the difference between the two is only evident when viewing the phase response.  We observe the prescense of a pole and zero which are on opposite sides of the imaginary axis. 

\begin{figure}[ht]
    \includegraphics[width=\linewidth]{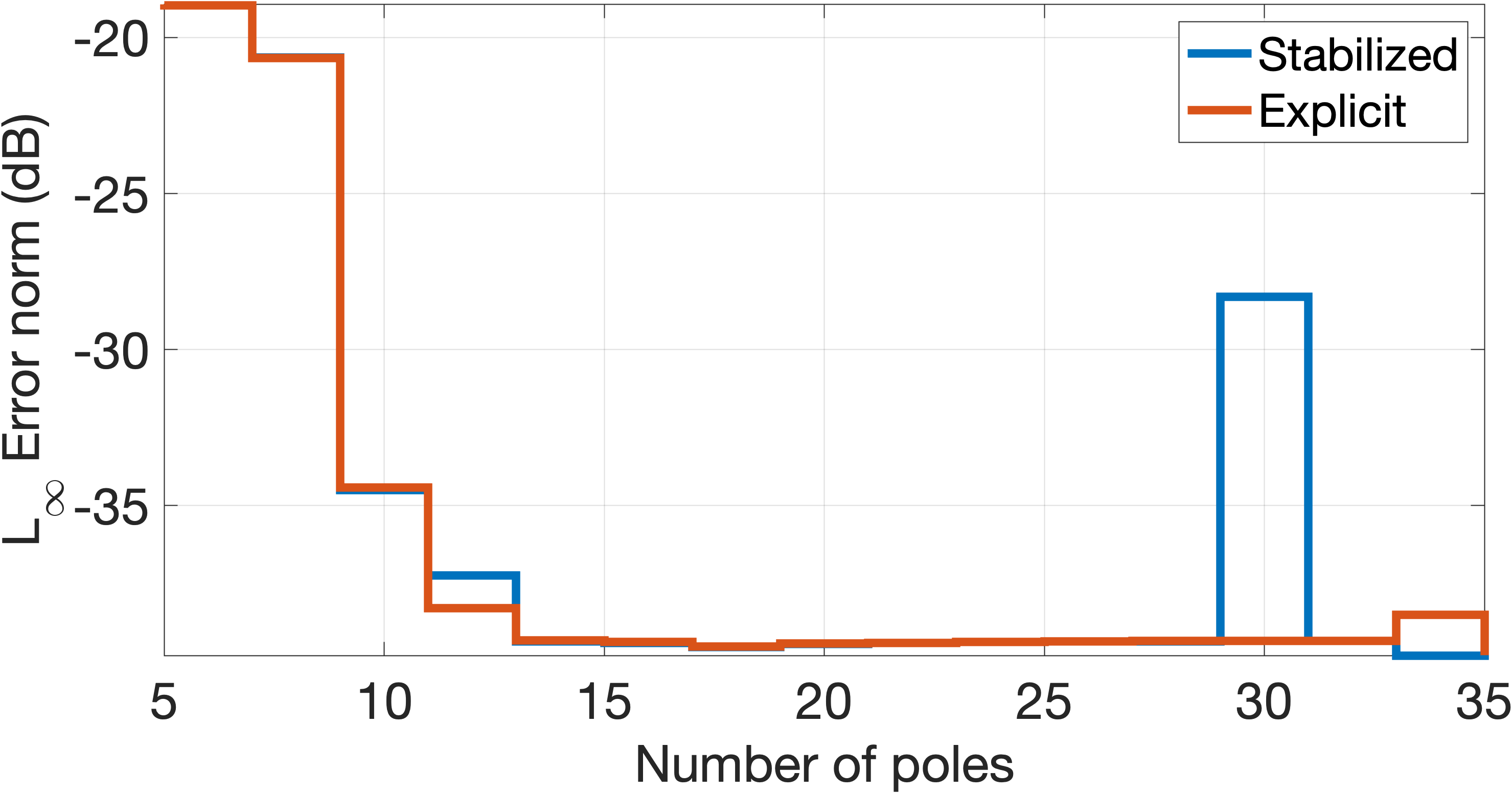}
    \caption{\capsize A plot showing the \(L_\infty\) norm of the error system, i.e. \(R-G\), against the number of poles of each system \(R\) generated using two different optimization approaches with the input ISS system \(G\).  The blue line and red line indicate the error norm of systems generated with the stability-enforced optimization problem and the explicit/unconstrained optimization problem respectively.}
    \label{fig:norm_opt}
\end{figure}

\begin{figure}[ht]
    \includegraphics[width=\linewidth]{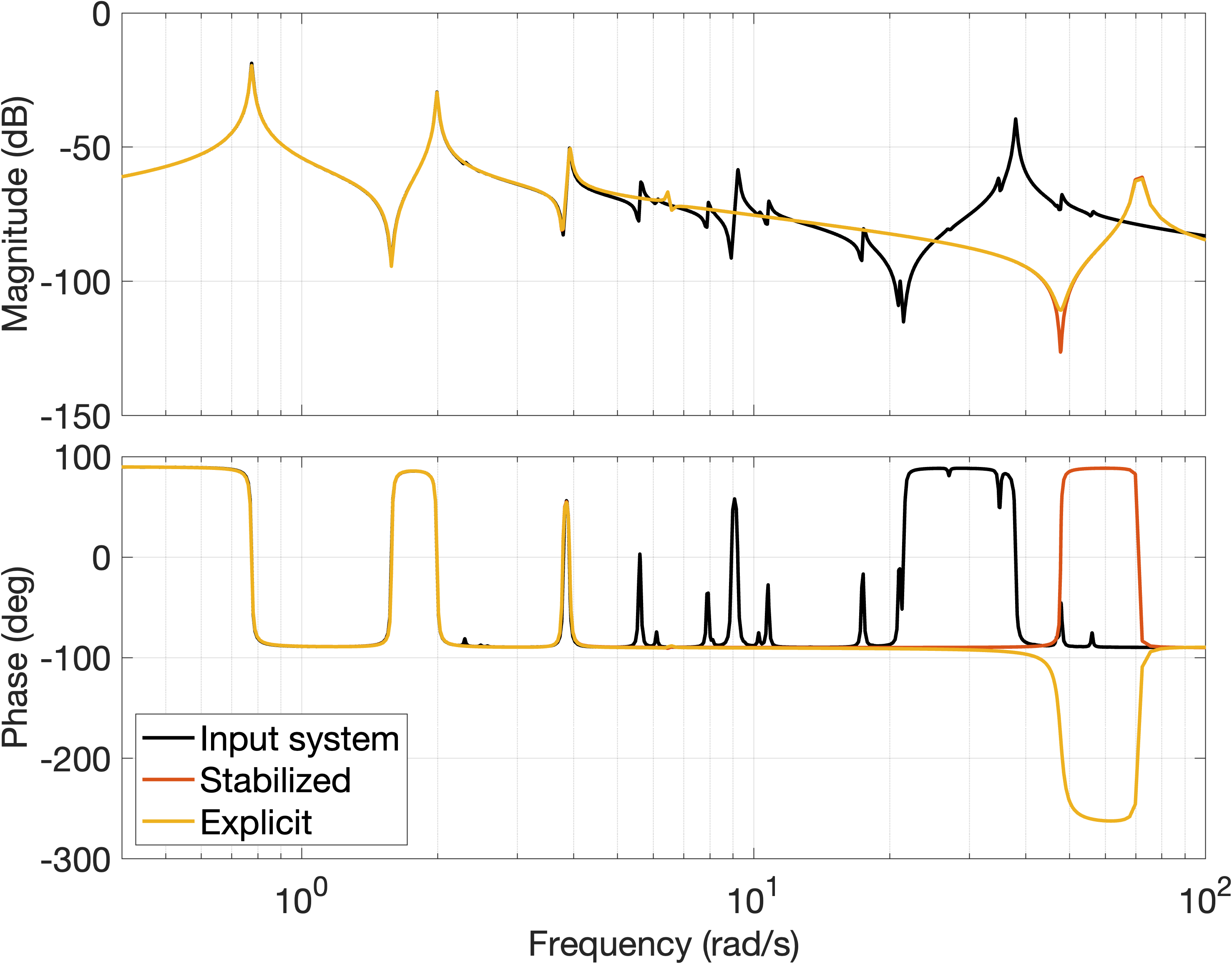}
    \caption{\capsize A Bode plot of the ISS system in black and two 11 order systems generated using the adaptive approach with the stability constrained optimization in red and the unconstrained optimization in yellow.}
    \label{fig:bode_stability}
\end{figure}

\section{Conclusion}
In this paper we introduced an iterative algorithm suitable for system identification on lightly-damped systems which requires a minimal number of system identification experiments.  The resulting system, which takes the form of a barycentric interpolant,  obtains its weight free parameters by solving an optimization problem that minimizes a proxy for the mean squared error between the unknown system's and the reconstruction system's output subject to sinusoidal inputs of varying frequencies; the minimizer of this optimization problem has an explicit expression which is the solution of a linear equation.  We also derived a relaxed optimization problem with an added stability constraint which can be solved with an LMI solver.  This relaxed optimization problem negligibly affects the performance of the algorithm in comparison to the unconstrained problem in general.  We explored two approaches to frequency selection: a gridded approach and an adaptive approach, and showed through our computational results that the adaptive approach yields more consistent and better-approximating systems.  On the whole, the algorithm produces well-performing, stable systems with a low number of system identification experiments performed.  

In future works we will investigate different adaptive frequency selection strategies that utilize the transient response data in place of the adaptive gridding strategy.   We will also rigorously prove various properties and performance metric bounds for the performance of the algorithm and its resulting systems.  Finally, we plan to test the algorithm on a lightly-damped system such as the Rikje tube and evaluate its performance in comparison to other system identification schemes.

\begin{appendix} 

\subsection{Detecting the Onset of Steady State} \label{transient_steady.app}
To estimate the frequency response of the system at a frequency \(\omega\), we will interrogate it with a sinusoidal wave with amplitude \(A\) and collect data until it reaches steady state.  Suppose we are sampling at a frequency of \(f_s\), and we are collecting the output of the system \(y\) in real time in chunks of length \(L\).  We define \(L\) such that it contains a sufficient number of complete cycles of the sinusoidal wave.  

Consider a chunk of data \(y_i\) from \(i=n\) to \(n+L\).  We will use least squares to check goodness of fit with a sinusoidal wave in order to determine whether the system has reached steady state.  Define the signals
\[c_i = \cos(\omega i/f_s), \; s_i = \sin(\omega i/f_s),\]
then estimate the coefficients \(x_1\) and \(x_2\) in the linear equation 
\[
	\begin{bmatrix}y_n & \cdots & y_{n+L}\end{bmatrix} = \begin{bmatrix}x_1 & x_2\end{bmatrix} 
		\begin{bmatrix}c_n & \cdots & c_{n+L} \\ s_n & \cdots & s_{n+L}\end{bmatrix}
\]
using least squares.  We now check the goodness of fit by measuring the residuals, i.e. \(\curly{r_i} = x_1 c_i + x_2 s_i - y_i\).  If the maximum relative error in the block 
\[
	\hat{\gamma} = \frac{\max_{i\in [n, \, n+L]} \left|r_i\right|}
				{\max_{i\in [n, \, n+L]} \left| y_i\right|}
\]
is less than some threshold \(\gamma\), then the goodness of fit is satisfactory and we can conclude that steady state has been reached. Finally we record the frequency response data as \(G(j\omega) \approx A(x_1 - jx_2)\) and save the transient data for later processing.  Otherwise, we wait until the next chunk of data is received and repeat the process.

\end{appendix}

\bibliographystyle{IEEEtran}
\bibliography{IEEEabrv, paper_bib}

\begin{thebibliography}{10}
\providecommand{\url}[1]{#1}
\csname url@rmstyle\endcsname
\providecommand{\newblock}{\relax}
\providecommand{\bibinfo}[2]{#2}
\providecommand\BIBentrySTDinterwordspacing{\spaceskip=0pt\relax}
\providecommand\BIBentryALTinterwordstretchfactor{4}
\providecommand\BIBentryALTinterwordspacing{\spaceskip=\fontdimen2\font plus
\BIBentryALTinterwordstretchfactor\fontdimen3\font minus \fontdimen4\font\relax}
\providecommand\BIBforeignlanguage[2]{{%
\expandafter\ifx\csname l@#1\endcsname\relax
\typeout{** WARNING: IEEEtran.bst: No hyphenation pattern has been}%
\typeout{** loaded for the language `#1'. Using the pattern for}%
\typeout{** the default language instead.}%
\else
\language=\csname l@#1\endcsname
\fi
#2}}

\bibitem{gustavsen2002rational}
B.~Gustavsen and A.~Semlyen, ``Rational approximation of frequency domain responses by vector fitting,'' \emph{IEEE Transactions on power delivery}, vol.~14, no.~3, pp. 1052--1061, 2002.

\bibitem{sanathanan2003transfer}
C.~Sanathanan and J.~Koerner, ``Transfer function synthesis as a ratio of two complex polynomials,'' \emph{IEEE transactions on automatic control}, vol.~8, no.~1, pp. 56--58, 2003.

\bibitem{jamaludin2013n4sid}
I.~Jamaludin, N.~Wahab, N.~Khalid, S.~Sahlan, Z.~Ibrahim, and M.~F. Rahmat, ``N4sid and moesp subspace identification methods,'' in \emph{2013 IEEE 9th international colloquium on signal processing and its applications}.\hskip 1em plus 0.5em minus 0.4em\relax IEEE, 2013, pp. 140--145.

\bibitem{astolfi2010model}
A.~Astolfi, ``Model reduction by moment matching for linear and nonlinear systems,'' \emph{IEEE Transactions on Automatic Control}, vol.~55, no.~10, pp. 2321--2336, 2010.

\bibitem{ionita2014data}
A.~C. Ionita and A.~C. Antoulas, ``Data-driven parametrized model reduction in the loewner framework,'' \emph{SIAM Journal on Scientific Computing}, vol.~36, no.~3, pp. A984--A1007, 2014.

\bibitem{antoulas2017tutorial}
A.~C. Antoulas, S.~Lefteriu, A.~C. Ionita, P.~Benner, and A.~Cohen, ``A tutorial introduction to the loewner framework for model reduction,'' \emph{Model Reduction and Approximation: Theory and Algorithms}, vol.~15, p. 335, 2017.

\bibitem{karachalios2021loewner}
D.~Karachalios, I.~V. Gosea, and A.~C. Antoulas, ``The loewner framework for system identification and reduction,'' in \emph{Model Order Reduction: Volume I: System-and Data-Driven Methods and Algorithms}.\hskip 1em plus 0.5em minus 0.4em\relax de Gruyter, 2021, pp. 181--228.

\bibitem{Epperlein15}
J.~P. Epperlein, B.~Bamieh, and K.~J. Astrom, ``Thermoacoustics and the rijke tube: Experiments, identification, and modeling,'' \emph{IEEE Control Systems Magazine}, vol.~35, no.~2, pp. 57--77, 2015.

\bibitem{Nakatsukasa_2018}
Y.~Nakatsukasa, O.~S{\`{e}}te, and L.~N. Trefethen, ``The {AAA} algorithm for rational approximation,'' \emph{{SIAM} Journal on Scientific Computing}, vol.~40, no.~3, pp. A1494--A1522, jan 2018.

\bibitem{Gosea24Stable}
T.~Bradde, S.~Grivet-Talocia, Q.~Aumann, and I.~V. Gosea, ``A modified aaa algorithm for learning stable reduced-order models from data,'' 2024.

\bibitem{benner2021interpolation}
P.~Benner and P.~Goyal, ``Interpolation-based model order reduction for polynomial systems,'' \emph{SIAM Journal on Scientific Computing}, vol.~43, no.~1, pp. A84--A108, 2021.

\bibitem{Jonas2024}
J.~Jonas and B.~Bamieh, ``An adaptation of the {AAA}-interpolation algorithm for model reduction of {MIMO} systems,'' in \emph{2024 American Control Conference (ACC)}, ser. 2024 American Control Conference (ACC).\hskip 1em plus 0.5em minus 0.4em\relax IEEE, 2024.

\bibitem{Jonas2025TAC}
------, ``An iterative tangential interpolation framework for model reduction of {MIMO} systems,'' 2025.

\bibitem{Caverly24}
R.~J. Caverly and J.~R. Forbes, ``Lmi properties and applications in systems, stability, and control theory,'' 2024.

\bibitem{iss_model}
\BIBentryALTinterwordspacing
Y.~Chahlaoui, ``Benchmark examples for model reduction.'' [Online]. Available: \url{http://slicot.org/20-site/126-benchmark-examples-for-model-reduction}
\BIBentrySTDinterwordspacing

\end{thebibliography}

\end{document}